\documentclass[a4paper,UKenglish,cleveref, autoref, thm-restate,authorcolumns]{lipics-v2019}

\title{Better Bounds on the Adaptivity Gap of Influence Maximization under Full-adoption Feedback} 

\titlerunning{Adaptivity Gap in Influence Maximization Problems}
\author{Gianlorenzo D'Angelo}{Gran Sasso Science Institute, Italy}{gianlorenzo.dangelo@gssi.it}{}{}
\author{Debashmita Poddar}{Gran Sasso Science Institute, Italy}{debashmita.poddar@gssi.it}{}{}
\author{Cosimo Vinci}{Gran Sasso Science Institute, Italy}{cosimo.vinci@gssi.it}{}{}
\authorrunning{G. D'Angelo et al.} 
\Copyright{Gianlorenzo D'Angelo, Debashmita Poddar, and Cosimo Vinci} 


\category{} 

\relatedversion{} 

\supplement{}

\acknowledgements{This work has been partially supported by the Italian MIUR PRIN 2017 Project ALGADIMAR \lq\lq Algorithms, Games, and Digital Markets\rq\rq.}

\nolinenumbers 

\EventEditors{John Q. Open and Joan R. Access}
\EventNoEds{2}
\EventLongTitle{42nd Conference on Very Important Topics (CVIT 2016)}
\EventShortTitle{CVIT 2016}
\EventAcronym{CVIT}
\EventYear{2016}
\EventDate{December 24--27, 2016}
\EventLocation{Little Whinging, United Kingdom}
\EventLogo{}
\SeriesVolume{42}
\ArticleNo{23}
\usepackage{amsmath,amsfonts,amssymb,bbm,bm}
\usepackage{algorithm}
\usepackage{algorithmic}
\newcommand{\E}{\mathbb{E}}

\begin{document}

\maketitle
\hideLIPIcs 
\begin{abstract}
In the \emph{influence maximization (IM)} problem, we are given a social network and a budget $k$, and we look for a set of $k$ nodes in the network, called seeds, that maximize the expected number of nodes that are reached by an influence cascade generated by the seeds, according to some stochastic model for influence diffusion.
Extensive studies have been done on the IM problem, since his definition by Kempe et al.~\cite{Kempe2003}. However, most of the work focuses on the \emph{non-adaptive} version of the problem where all the $k$ seed nodes must be selected before that the cascade starts. In this paper, we study the \emph{adaptive} IM, where the nodes are selected sequentially one by one, and the decision on the $i$th seed can be based on the \emph{observed} cascade produced by the first $i-1$ seeds. We focus on the \emph{full-adoption feedback} in which we can observe the entire cascade of each previously selected seed and on the independent cascade model where each edge is associated with an independent probability of diffusing influence.

Previous works showed that there are constant upper bounds on the adaptivity gap, which compares the performance of an adaptive algorithm against a non-adaptive one, but the analyses used to prove these bounds only works for specific graph classes such as in-arborescences, out-arborescences, and one-directional bipartite graphs. Our main result is the first sub-linear upper bound that holds for any graph. Specifically, we show that the adaptivity gap is upper-bounded by $\lceil n^{1/3}\rceil $, where $n$ is the number of nodes in the graph. Moreover, we improve over the known upper bound for in-arborescences from $\frac{2e}{e-1}\approx 3.16$ to $\frac{2e^2}{e^2-1}\approx 2.31$. Finally, we study $\alpha$-bounded graphs, a class of undirected graphs in which the sum of node degrees higher than two is at most $\alpha$, and show that the adaptivity gap is upper-bounded by $\sqrt{\alpha}+O(1)$. Moreover, we show that in 0-bounded graphs, i.e. undirected graphs in which each connected component is a path or a cycle, the adaptivity gap is at most $\frac{3e^3}{e^3-1}\approx 3.16$.
To prove our bounds, we introduce new techniques to relate adaptive policies with non-adaptive ones that might be of their own interest.
\end{abstract}

\section{Introduction}

In the Influence Maximization (IM) problem, we are given a social network, a stochastic model for diffusion of influence over the network, and a budget $k$, and we ask to find a set of $k$ nodes, called \emph{seeds}, that maximize their \emph{spread of influence}, which is the expected number of nodes reached by a cascade of influence diffusion generated by the seeds according to the given diffusion model.
One of the most studied model for influence diffusion is the Independent Cascade (IC), where each edge is associated with an independent probability of transmitting influence from the source node to the tail node. In the IC model the spread of influence is a monotone submodular function of the seed set, therefore a greedy algorithm guarantees a $1-\frac{1}{e}$ approximation factor for the IM problem~\cite{Kempe2015a}.
Since his definition by Domingos and Richardson~\cite{Domingos2001,Richardson2002} and formalization as an optimization problem by Kempe et al.~\cite{Kempe2003,Kempe2015a}, the IM problem and its variants have been extensively investigated, motivated by applications in viral marketing~\cite{Chen10}, adoption of technological innovations~\cite{Goldberg2013}, and outbreak or failure detection~\cite{Leskovec2007}. See~\cite{DBLP:series/synthesis/2013Chen,Li2018} for surveys on the IM problem.

Recently, Golovin and Krause~\cite{Golovin2011a} initiated the study of the IM problem under the framework of adaptive optimization, where, instead of selecting all the seeds at once at the beginning of the process, we can select one seed at a time and observe, to some extent, the portion of the network reached by a new selected seed. The advantage is that the decision on the next seed to choose can be based on the observed spread of previously selected seeds, usually called \emph{feedback}.
Two main feedback models have been introduced: in the \emph{full-adoption} feedback the whole spread from each seed can be observed, while in the \emph{myopic} feedback one can only observe the direct neighbors of each seed.

Golovin and Krause considered the Independent Cascade model and showed that, under full-adoption feedback, the objective function satisfies the property of \emph{adaptive submodularity} (introduced in the same paper) and therefore a greedy algorithm achieves a $1-\frac{1}{e}$ approximation for the adaptive IM problem. They also conjectured that there exists a constant factor approximation algorithm for the myopic feedback model, which has been indeed found by Peng and Chen~\cite{Peng2019} who proposed a $\frac{1}{4}\left(1-\frac{1}{e}\right)$-approximation algorithm.

However, the approximation ratio for the adaptive IM problem, which compares a given adaptive algorithm with an optimal adaptive one, does not measure the benefits of implementing adaptive policies over non-adaptive ones. 
To this aim, Chen and Peng~\cite{Chen2019,Peng2019} introduced the adaptivity gap, which is the supremum, over all possible inputs, of the ratio between the spread of an optimal adaptive policy and that of an optimal non-adaptive one. 
In~\cite{Peng2019}, Peng and Chen considered independent cascade model with myopic feedback and showed that the adaptivity gap is between $\frac{e}{e-1}$ and 4 for any graph.
In~\cite{Chen2019}, the same authors showed some upper and lower bounds on the adaptivity gap in the case of full-adoption feedback, still under independent cascade, for some particular graph classes. Specifically, they showed that the adaptivity gap is in the interval $\left[ \frac{e}{e-1}, \frac{2e}{e-1}\right]$ for in-arborescences and it is in the interval $\left[ \frac{e}{e-1},2\right]$ for out-arborescences. Moreover, it is equal to $\frac{e}{e-1}$ in one-directional bibartite graphs. In order to show these bounds, they followed an approach introduced by Asadpour and Nazerzadeh~\cite{Asadpour16} which consists in transforming an adaptive policy into a non-adaptive one by means of multilinear extensions, and constructing a Poisson process to relate the influence spread of the non-adaptive policy to that of the adaptive one. For general graphs and full-adoption feedback, the only known upper bounds on the adaptivity gap are linear in the size of the graph and can be trivially derived.
%

In this paper, we consider the independent cascade model with full-adoption feedback, and show the first sub-linear upper bound on the adaptivity gap that holds for general graphs. In detail we show that that the adaptivity gap is at most $\lceil n^{1/3}\rceil$, where $n$ is the number of nodes in the graph. Moreover, we tighten the upper bound on the adaptivity gap for in-arborescences by showing that it is at most $\frac{2e^2}{e^2-1}<\frac{2e}{e-1}$.
Using similar techniques we study the adaptivity gap of \emph{$\alpha$-bounded graphs}, which is the class of undirected graphs where the sum of node degrees higher than two is at most $\alpha$.
We show that the adaptivity gap is upper-bounded by $\sqrt{\alpha}+O(1)$, which is smaller that $O(n^{1/3})$ for several graph classes. 
In 0-bounded graphs, i.e. undirected graphs in which each connected component is a path or a cycle, the adaptivity gap is at most $\frac{3e^3}{e^3-1}$.
To prove our bounds, we introduce new techniques to connect adaptive policies with non-adaptive ones that might be of their own interest.

\subsection*{Related Work}
\subparagraph*{Influence Maximization.} Several studies based on general graphs \cite{Lowalekar2016, 7403743,Schoenebeck2019, Tang2014} have been conducted since the seminal paper by Kempe et al.~\cite{Kempe2015a}. Schoenebeck and Tao~\cite{Schoenebeck2019} studied the influence maximization problem on undirected graphs and proves that it is APX-hard for both the independent cascade and the linear threshold problem. 
Borgs et al.~\cite{Borgs2014} propose an efficient algorithm that runs in quasilinear time and still guarantees an approximation factor of $1-\frac{1}{e}-\epsilon$, for any $\epsilon>0$.
Tang et al.~\cite{Tang2014} propose an algorithm which is experimentally close to the optimal one under the independent cascade model. Mihara et al.~\cite{7403743} consider unknown graphs for the influence maximization problem and devises an algorithm which achieves a fraction between 0.6 and 0.9 of the influence spread with minimal knowledge of the graph topology. Extensive literature reviews on influence maximization and its machinery is provided by Chen et al.~\cite{DBLP:series/synthesis/2013Chen} and Li et al.~\cite{Li2018}.

Several works on the adaptive influence maximization problem~\cite{Han2018a,Sun2018, Tang2019, Tong2019, Tong2017,  DBLP:journals/corr/VaswaniL16, Yuan2017} evolved after the concept introduced by Golovin and Krause~\cite{Golovin2011a}, and explore the adaptive optimization under different feedback models. The myopic model (in which, one can only observe the nodes influenced by the seed nodes) has been studied  in~\cite{ Peng2019,Salha2018}. Sun et al.~\cite{Sun2018} capture the scenario in which, instead of considering one round, the diffusion process takes over $T$ rounds, and a seed set of at most $k$ nodes is selected at each round. The authors designed a greedy approximation algorithm that guarantees a constant approximation ratio. Tong and Wang~\cite{Tong2015} introduce a new version of the adaptive influence maximization problem by adding a time constraint. Other than the classic full-adoption and myopic feedback model, Yuan and Tang~\cite{Yuan2017}, and Tong and Wang~\cite{Tong2019}, have also introduced different feedback models that use different parameters to overcome the need of submodularity to guarantee a good approximation. Han et al.~\cite{Han2018a} propose a framework which uses existing non-adaptive techniques to construct a strong approximation for a generalization of the adaptive influence maximization problem in which in each step a batch of node is selected.

\subparagraph*{Adaptivity Gaps.} Adaptivity gaps for the problem of maximizing stochastic monotone submodular functions have been studied by Asadpour and
Nazerzadeh~\cite{Asadpour16}.
A series of work studied adaptivity gaps for a two-step adaptive influence maximization problem~\cite{Badanidiyuru2016,Rubinstein2015, Seeman2013, Singer2016}. Gupta et al.~\cite{Jiang2020, Gupta2017} worked on the adaptivity gaps for stochastic probing. A recent line of studies has been conducted~\cite{Chen2019, Chen2019a, Peng2019} which focuses on finding the adaptivity gaps on different graph classes using the classical feedback models. Peng and Chen~\cite{Peng2019} confirmed a conjecture of Golovin and Krause~\cite{Golovin2011a}, which states that the adaptive greedy algorithm with myopic feedback is a constant approximation of the adaptive optimal solution. They show that the adaptivity gap of the independent cascade model with myopic feedback belongs to $[\frac{e}{e-1}, 4]$. Chen et al.~\cite{Chen2019a} introduced the greedy adaptivity gap, which compares the performance of the adaptive and the non-adaptive greedy algorithms. They show that the infimum of the greedy adaptivity gap is $1-\frac{1}{e}$ for every combination of diffusion and feedback models.
The most related work to our results is that of~\cite{Chen2019}. Chen and Peng~\cite{Chen2019} derive upper and lower bounds on the adaptivity gap under the independent cascade model with full-adoption feedback, when the considered graphs are in-arborescences, out-arborescences, or one-directional bipartite graphs. In particular, they show that the adaptivity gaps of in-arborescences and out-arborescences are in the intervals $\left[ \frac{e}{e-1}, \frac{2e}{e-1}\right]$ and $\left[ \frac{e}{e-1}, 2\right]$, respectively, and they provide a tight bound of $\frac{e}{e-1}$ on the adaptivity gap of one-directional bipartite graphs.


\subsection*{Organization of the Paper}
In Section \ref{sec_prel} we give the preliminary definitions and notations which this work is based on. Sections~\ref{sec_inarb}--\ref{sec_other} are devoted to the main technical contribution of the paper (i.e., adaptivity gaps of in-arborescences, general graphs, and $\alpha$-bounded graphs). In Section \ref{sec_future}, we highlight some future research directions. Due to the lack of space, some missing proofs are deferred to the appendix. 

\section{Preliminaries}\label{sec_prel}
For two integers $h$ and $k$, $h\leq k$, let $[k]_h:=\{h,h+1,\ldots, k\}$ and $[k]:=[k]_1$. 

\subparagraph*{Independent Cascade Model.} In the {\em independent cascade model} (IC), we have an {\em influence graph}  $G=(V,E,(p_{uv})_{(u,v)\in E})$, where $p_{uv}\in [0,1]$ is an {\em activation probability} associated to each edge $(u,v)\in E$. Given a set of {\em seed nodes} $S\subseteq V$ which are initially \emph{active}, the diffusion process in the IC model is defined in $t\geq 0$  discrete steps as follows: (i) let $A_t$ be the set of active nodes which are activated at each step $t\geq 0$; (ii) $A_0:=S$; (iii) given a step $t\geq 0$, for any edge $(u,v)$ such that $u\in A_t$, node $u$ can activate node $v$ with probability $p_{uv}$ independently from any other node, and, in case of success, $v$ is included in $A_{t+1}$; (iv) the diffusion process ends at a step $r\geq 0$ such that $A_{r}=\emptyset$, i.e., no node can be activated at all, and $\bigcup_{t\leq r} {A_t}$ is the {\em influence spread}, i.e., the set of nodes activated/reached by the diffusion process. 

The above diffusion process can be equivalently defined as follows. The {\em live-edge graph}  $L=(V,L(E))$ of $G$ is a random graph made from $G$, where $L(E)\subseteq E$  is a subset of edges such that each edge $(u,v)\in E$ is included in $L(E)$ with probability $p_{uv}$, independently from the other edges. Given a live-edge graph $L$, let $R(S,L):=\{v\in V:\text{ there exists a path from $u$ to $v$ in $L$ for some $u\in S$}\}$, i.e., the set of nodes reached by nodes in $S$ in the live-edge graph $L$. Informally, if $S$ is the set of seed nodes, and $L$ is a live-edge graph, $R(S,L)$ equivalently denotes the set of nodes which are reached/activated by the above diffusion process. 
Given a set of seed nodes $S$, the {\em expected influence spread} of $S$ is defined as $\sigma(S):=\mathbb{E}_L[|R(S,L)|]$. 

\subparagraph*{Non-adaptive Influence Maximization.} The {\em non-adaptive influence maximization problem under the IC model} is the computational problem that, given an influence graph $G$ and an integer $k\geq 1$, asks to find a set of seed nodes $S\subseteq V$ with $|S|=k$ such that $\sigma(S)$ is maximized. 

\subparagraph*{Adaptive Influence Maximization.} Differently from the non-adaptive setting, in which all the seed nodes are activated at the beginning and then the influence spread is observed, an {\em adaptive policy} activates the seeds sequentially in $k$ steps,
one seed node at each step, and the decision on the next seed node to select is based on the feedback resulting from the observed spread of previously selected nodes. The feedback model considered in this work is {\em full-adoption}: when a node is selected, the adaptive policy observes its entire influence spread. 

An adaptive policy under the full-adoption feedback model is formally defined as follows. Given a live-edge graph $L$, the {\em realisation} $\phi_L:V\rightarrow 2^V$ associated to $L$ assigns to each node $v\in V$ the value $\phi_L(v):=R(\{v\},L)$, i.e., the set of nodes activated by $v$ under a live-edge graph $L$. Given a set $S\subseteq V$, a {\em partial realisation} $\psi:S\rightarrow 2^V$ is the restriction to $S$ of some realisation, i.e., there exists a live-edge graph $L$ such that $\psi(v)=\phi_L(v)$ for any $v\in S$. Given a partial realisation $\psi:S\rightarrow 2^V$, let $dom(\psi):=S$, i.e., $dom(\psi)$ is the domain of  partial realisation $\psi$, let $R(\psi):=\bigcup_{v\in S}\psi(v)$, i.e., $R(\psi)$ is the set of nodes reached/activated by the diffusion process when the set of seed nodes is $S$, and let $f(\psi):=|R(\psi)|$. A partial realisation $\psi'$ is a {\em sub-realisation} of $\psi$ (or, equivalently,  $\psi'\subseteq \psi$), if $dom(\psi')\subseteq dom(\psi)$ and $\psi'(v)=\psi(v)$ for any $v\in dom(\psi')$. We observe that a partial realisation $\psi$ can be equivalently represented as $\{(v,R(\{v\},L)):v\in dom(\psi)\}$ for some live-edge graph $L$. 

An adaptive policy $\pi$ takes as input a partial realisation $\psi$ and, either returns a node $\pi(\psi)\in V$ and activates it as seed, or interrupts the activation of new seed nodes, e.g., by returning a string $\pi(\psi):=STOP$. An adaptive policy $\pi$ can be run as in Algorithm \ref{ad_alg}. 
\begin{algorithm}[ht]
	\caption{Adaptive algorithm}
	\label{ad_alg}
	\begin{algorithmic}[1]
		\REQUIRE an influence graph $G$ and an adaptive policy $\pi$;
		\ENSURE a partial realisation;
		\STATE let $L$ be the live-edge graph;
		\STATE let $\psi:=\emptyset$ (i.e., $\psi$ is the empty partial realisation);
		\WHILE{$\pi(\psi)\neq STOP$}
		\STATE $v:=\pi(\psi)$;
		\STATE $\psi:=\psi\cup \{(v,R(\{v\},L))\}$;
		\ENDWHILE
		\RETURN $\psi_{\pi,L}:=\psi$;
	\end{algorithmic}
\end{algorithm}

The {\em expected influence spread} of an adaptive policy $\pi$ is defined as $\sigma(\pi):=\mathbb{E}_L[f(\psi_{\pi,L})]$, i.e., it is the expected value (taken on all the possible live-edge graphs) of the number of nodes reached by the diffusion process at the end of Algorithm \ref{ad_alg}. We say that $|\pi|=k$ if policy $\pi$ always return a partial realisation $\psi_{\pi,L}$ with $|dom(\psi_{\pi,L})|=k$. The {\em adaptive influence maximization problem under the IC model} is the computational problem that, given an influence graph $G$ and an integer $k\geq 1$, asks to find an adaptive policy $\pi$ that maximizes the expected influence spread $\sigma(\pi)$ subject to constraint $|\pi|=k$. 

\subparagraph*{Adaptivity gap.}
Given an influence graph $G$ and an integer $k\geq 1$, let $OPT_N(G,k)$ (resp. $OPT_A(G,k)$) denote the optimal value of the non-adaptive (resp. adaptive) influence maximization problem with input $G$ and $k$. Given a class of influence graphs $\mathcal{G}$ and an integer $k\geq 1$, the {\em $k$-adaptivity gap} of $\mathcal G$ is defined as $$AG(\mathcal{G},k):=\sup_{G\in\mathcal{G}}\frac{OPT_A(G,k)}{OPT_N(G,k)},$$ and measures how much an adaptive policy outperforms a non-adaptive solution for the influence maximization problem applied to influence graphs in $\mathcal{G}$, when the maximum number of seed nodes is $k$. The {\em adaptivity gap} of $\mathcal{G}$ is defined as $AG(\mathcal{G}):=\sup_{k\geq 1}AG(\mathcal{G},k)$. We observe that for $k=1$ or $n\leq k$ the $k$-adaptivity gap is trivially equal to 1, thus we omit such cases in the following. 
\section{Adaptivity Gap for In-arborescences}\label{sec_inarb}
An {\em in-arborescence} is a graph $G=(V,E)$ that can be constructed from a rooted tree $T=(V,F)$, by adding in $E$ an edge $(v,u)$ if $u$ is a father of $v$ in tree $T$. An upper bound of $\frac{2e}{e-1}\approx 3.16$ on the adaptivity gap of in-arborescences has been provided in \cite{Chen2019}. In this section we provide an improved upper bound for such graphs. 
\begin{theorem}\label{thm1}
If $\mathcal{G}$ is the class of all the in-arborescences, then $$AG(\mathcal{G},k)\leq \frac{2}{1-(1-2/k)^k}\leq \frac{2e^2}{e^{2}-1}\approx 2.31,\ \forall k\geq 2.$$
\end{theorem}

Let $G=(V=[n],E,(p_{uv})_{(u,v)\in E})$ be an in-arborescence, where $n>k$ is the number of nodes. To show the claim of Theorem \ref{thm1}, we need some preliminary notations and lemmas. Given a partial realisation $\psi$, and a node $i\in [n]$, let 
$$\Delta(i|\psi):=\E_L[f(\psi\cup \{(i,R(\{i\},L))\})-f(\psi)|\psi\subseteq \phi_L],$$
i.e., $\Delta(i|\psi)$ is the expected increment of the influence spread due to node $i$ when the observed partial realisation is $\psi$. We have the following claim (from \cite{Golovin2011a}), holding even for general graphs, whose proof is trivial. 
\begin{claim}[Adaptive Submodularity, \cite{Golovin2011a}]\label{lem0}
Let $G$ be an arbitrary influence graph. For any partial realisations $\psi,\psi'$ of $G$ such that $\psi\subseteq \psi'$, and any node $i\notin R(\psi')$, we have that  $\Delta(i|\psi')\leq \Delta(i|\psi)$. 
\end{claim}

An adaptive policy $\pi$ is called {\em randomized} if, for any partial realisation $\psi$, node $\pi(\psi)$ is not selected deterministically (in general), but randomly (according to a probability distribution $p_{\psi}$ depending on $\psi$). Given a vector $\bm y=(y_1,\ldots, y_n)$ such that $y_i\in [0,1]$ for any $i\in [n]$, we say that $\mathbb{P}(\pi)=\bm y$ if the probability that each node $i$ belongs to $dom(\psi_{\pi,L})$ is $y_i$, where $\psi_{\pi,L}$ is the partial realisation returned by Algorithm \ref{ad_alg} with policy $\pi$. Let $OPT_A(G,\bm y)$ be the optimal expected influence spread $\sigma(\pi)$ over all the randomized adaptive policies $\pi$ such that $\mathbb{P}(\pi)=\bm y$.\footnote{We observe that, if $\bm y$ is arbitrary, a deterministic policy $\pi$ verifying $\mathbb{P}(\pi)=\bm y$ might not exists, and the introduction of randomization solves this issue.}

Let $\pi^*$ be an optimal adaptive policy for the adaptive influence maximization problem (with $|\pi^*|=k$), and let $\bm x=(x_1,\ldots, x_n)$ be the vector such that $\mathbb{P}(\pi^*)=\bm x$. As $|\pi^*|=k$, we have that $\sum_{i\in [n]} x_i=k$. 

For any $t\in [k]_0$, let $S_t$ be the optimal set of $t$ seed nodes in the non-adaptive influence maximization problem, i.e., such that $OPT_N(G,t)=\mathbb{E}_L(|R(S_t,L)|)$. Let $\psi_{t,L}$ be the random variable denoting the sub-realisation of $\phi_L$ such that $dom(\psi_{t,L})=S_t$. Let $\rho$ be the random variable equal to node $i\in [n]$ with probability $x_i/k$. Observe that the above random variable is well-defined, as $\sum_{i\in [n]}(x_i/k)=k/k=1$. For any $t\in [k]$, let ${\psi}_{\rho,t,L}$ be the random variable denoting the sub-realisation of $\phi_L$ such that $dom({\psi}_{\rho,t,L})=S_{t-1}\cup\{\rho\}$.

We observe that ${\psi}_{\rho,t,L}$ is the partial realisation coming from the following  {\em hybrid non-adaptive policy}: initially, we activate the first $t-1$ seed nodes as in the optimal non-adaptive solution guaranteeing an expected influence spread of $OPT_N(G,t-1)$; then, we randomly choose a node $v$ according to random variable $\rho$ and we select $v$ as $t$-th seed node (if not already selected as seed). We use this hybrid non-adaptive policy as a main tool to obtain an improved upper bound on the adaptivity gap for in-arborescences. In the following lemma, holding even for general graphs, we relate the hybrid non-adaptive policy and the optimal non-adaptive solution, with the optimal adaptive policy. The proof structure exhibits some similarities with Lemma 6 of \cite{Asadpour16} and Lemma 3.3 of \cite{Chen2019}, but in their approach, they relate non-adaptive policies based on the Poisson process and multi-linear extensions, with the optimal adaptive policy. 
\begin{lemma}\label{lem1}
Let $G$ be an arbitrary influence graph. For any $t\in [k]$, and any fixed partial realisation $\psi$ of $G$ such that $\mathbb{P}[\psi_{t-1,L}=\psi]>0$, we have  
\begin{equation*}
\sigma(R(\psi))+k\cdot \E_{L,\rho}\left[f({\psi}_{\rho,t,L})-f(\psi_{t-1,L})|\psi_{t-1,L}=\psi \right]\geq OPT_A(G,k)
\end{equation*}
\end{lemma}
\begin{proof}
We have 
\begin{align}
&k\cdot \E_{L,\rho}\left[f({\psi}_{\rho,t,L})-f(\psi_{t-1,L})|\psi_{t-1,L}=\psi \right]\nonumber\\
=&k\cdot \sum_{i\in [n]}\mathbb{P}[\rho=i]\cdot \Delta(i|\psi)\nonumber\\
=&k\cdot \sum_{i\in [n]\setminus R(\psi)}\frac{x_i}{k}\cdot \Delta(i|\psi)\label{eq0}\\
=&\sum_{i\in [n]\setminus R(\psi)}x_i\cdot \Delta(i|\psi),\label{eq1}
\end{align}
where \eqref{eq0} holds since $\Delta(i|\psi)=0$ for any $i\in R(\psi)$. 

Let $\bm x'=(x_1',\ldots x_n')$ be the vector such that $x_i'=1$ if $i\in R(\psi)$, and $x_i'=x_i$ otherwise. As $x_i'\geq x_i$ for any $i\in [n]$ we have  
\begin{equation}\label{eq2}
OPT_A(G,k)\leq OPT_A(G,\bm x)\leq OPT_A(G,\bm x').
\end{equation}

Let $\pi'$ be the optimal randomized adaptive policy such that $\mathbb{P}(\pi')=\bm x'$. Policy $\pi'$ selects each node in $R(\psi)$ with probability $1$, thus we can assume that such seed nodes are selected at the beginning and that the adaptive policy starts by observing the resulting partial realisation. Furthermore, we can assume that, for any partial realisation $\psi'$, $\pi'$ does not select any node $i\in R(\psi')$, otherwise there is no increase of the influence spread. Given $j\in [n]$, let $\Delta'(j)$ denote the expected increment of the influence spread when $\pi'$ selects the $j$-th seed node (in order of selection, and without considering in the count the initial seeds of $R(\psi)$); analogously, let $\Delta'(j|i)$ denote the expected increment of the influence spread when $\pi'$ selects the $j$-th seed node, conditioned by the fact that the $j$-th seed is node $i$.\footnote{If an execution of $\pi'$ requires less than $j$ steps, we assume that the increase of the influence spread at step $j$(that contributes to the expected values $\Delta'(j)$ and $\Delta'(j|i)$) is null.} We get
\begin{align}
 &OPT_A(G,\bm x')\nonumber\\
  =&\sigma(R(\psi))+\sum_j \Delta'(j)\nonumber\\
  =&\sigma(R(\psi))+\sum_j \sum_{i\in [n]\setminus R(\psi)}\mathbb{P}[\text{the $j$-th seed node is $i$}]\cdot \Delta'(j|i)\nonumber\\
  =&\sigma(R(\psi))+\sum_{i\in [n]\setminus R(\psi)}\sum_j \mathbb{P}[\text{the $j$-th seed node is $i$}]\cdot \Delta'(j|i)\nonumber\\
    =&\sigma(R(\psi))+\sum_{i\in [n]\setminus R(\psi)}\sum_j \mathbb{P}[\text{the $j$-th seed node is $i$}]\cdot\nonumber\\
    &\ \ \ \cdot \E_{\pi'}[\Delta(i|\psi')|\text{$i=\pi'(\psi')$ for some $\psi'\supseteq \psi$ observed at step $j$}]\nonumber\\
        \leq &\sigma(R(\psi))+\sum_{i\in [n]\setminus R(\psi)}\sum_j \mathbb{P}[\text{the $j$-th seed node is $i$}]\cdot \Delta(i|\psi)\label{eq3.0}\\
                =&\sigma(R(\psi))+\sum_{i\in [n]\setminus R(\psi)}\mathbb{P}[\text{$i$ is selected as seed}]\cdot \Delta(i|\psi)\nonumber\\
=& \sigma(R(\psi))+\sum_{i\in [n]\setminus R(\psi)}x_i'\cdot \Delta(i|\psi)\nonumber\\
= & \sigma(R(\psi))+\sum_{i\in [n]\setminus R(\psi)}x_i\cdot \Delta(i|\psi),\label{eq3}
\end{align}
where \eqref{eq3.0} holds since $\Delta(i|\psi')\leq \Delta(i|\psi)$ for any partial realisation $\psi'\supseteq \psi$ by adaptive submodularity (Claim \ref{lem0}). By putting together \eqref{eq1}, \eqref{eq2}, and \eqref{eq3}, we get
\begin{align*}
& \sigma(R(\psi))+k\cdot \E_{L,\rho}\left[f({\psi}_{\rho,t,L})-f(\psi_{t-1,L})|\psi_{t-1,L}=\psi \right]\\
=& \sigma(R(\psi))+\sum_{i\in [n]\setminus R(\psi)}x_i\cdot \Delta(i|\psi)\\
\geq & OPT_A(G,\bm x')\\
\geq & OPT_A(G,k),
\end{align*}
thus showing the claim. 
\end{proof}

The following lemma is similar to Lemma 3.8 in~\cite{Chen2019} (for the proof, see the appendix). 

\begin{lemma}\label{lem2}
When the input influence graph $G$ is an in-arborescence, we have that 
\begin{equation*}
\sigma(R(\psi_{t-1,L}))\leq f(\psi_{t-1,L})+OPT_N(G,t-1)
\end{equation*}
for any $t\in [k]$ and any live-edge graph $L$. 
\end{lemma}

Armed with the above lemmas, we can now prove Theorem~\ref{thm1}.

\begin{proof}[Proof of Theorem~\ref{thm1}]
For any $t\in [k]$, we have
\begin{align}
&k\cdot (OPT_N(G,t)-OPT_N(G,t-1))\nonumber\\
=&k\cdot (\sigma(S_t)-\sigma(S_{t-1}))\nonumber\\
=&k\cdot (\E_L[f(\psi_{t,L})]-\E_L[f(\psi_{t-1,L})])\nonumber\\
\geq &k\cdot (\E_{L,\rho}[f({\psi}_{\rho,t,L})]-\E_L[f(\psi_{t-1,L})])\label{eqthm1.0}\\
= &k\cdot (\E_{L,\rho}[f({\psi}_{\rho,t,L})]-\E_{L,\rho}[f(\psi_{t-1,L})])\nonumber\\
=&k\cdot \E_{L,\rho}[f({\psi}_{\rho,t,L})-f(\psi_{t-1,L})]\nonumber\\
=& \E_{\psi_{t-1,L}}\left[k\cdot \E_{L,\rho}[f({\psi}_{\rho,t,L})-f(\psi_{t-1,L})|\psi_{t-1,L}]\right]\nonumber\\
\geq &\E_{\psi_{t-1,L}}[OPT_A(G,k)-\sigma(R(\psi_{t-1,L}))]\label{eqthm1}\\
\geq &\E_{\psi_{t-1,L}}[OPT_A(G,k)-f(\psi_{t-1,L})-OPT_N(G,t-1)]\label{eqthm3}\\
=& \E_{\psi_{t-1,L}}[OPT_A(G,k)]-\E_{\psi_{t-1,L}}[f(\psi_{t-1,L})]-\E_{\psi_{t-1,L}}[OPT_N(G,t-1)]\nonumber\\
= &OPT_A(G,k)-\sigma(S_{t-1})-OPT_N(G,t-1)\nonumber\\
= &OPT_A(G,k)-2\cdot OPT_N(G,t-1)\label{eqthm4},
\end{align}
where \eqref{eqthm1.0} holds since $dom(\psi_{t,L})$ is the optimal set of $t$ seed nodes for the non-adaptive influence maximization problem, \eqref{eqthm1} comes from Lemma \ref{lem1}, and \eqref{eqthm3} comes from Lemma \ref{lem2}. Thus, by \eqref{eqthm4}, we get $k\cdot (OPT_N(G,t)-OPT_N(G,t-1))\geq OPT_A(G,k)-2\cdot OPT_N(G,t-1)$, that after some manipulations leads to the following recursive relation:
\begin{equation}\label{fundeqthm}
OPT_N(G,t)\geq \frac{1}{k}\cdot OPT_A(G,k)+\left(1-\frac{2}{k}\right)\cdot OPT_N(G,t-1),\quad \forall t\in [k].
\end{equation}
By applying iteratively \eqref{fundeqthm}, we get
\begin{equation*}
OPT_N(G,k)\geq \frac{1}{k}\cdot \sum_{t=0}^{k-1}\left(1-\frac{2}{k}\right)^{t}\cdot OPT_A(G,k)=\frac{1-\left(1-2/k\right)^k}{2}\cdot OPT_A(G,k),
\end{equation*}
that leads to 
\begin{equation}
\frac{OPT_A(G,k)}{OPT_N(G,k)}\leq \frac{2}{1-(1-2/k)^k}\leq \frac{2}{1-e^{-2}} = \frac{2e^2}{e^{2}-1},
\end{equation}
and this shows the claim. 
\end{proof}
\section{Adaptivity Gap for General Influence  Graphs}\label{sec_gen}
In this section, we exhibit upper bounds on the $k$-adaptivity gap of general graphs. In the following theorem, we first give an upper bound that is linear in the number of seed nodes (see the appendix for the proof). 
\begin{theorem}\label{lemk}
Given an arbitrary class of influence graphs $\mathcal{G}$ and $k\geq 2$, we get $AG(\mathcal{G},k)\leq k$. 
\end{theorem}

In the next theorem we give an upper bound on the adaptivity gap that is sublinear in the number of nodes of the considered graph.

\begin{theorem}\label{thm2}
If $\mathcal{G}$ is the class of influence graphs having at most $n$ nodes, we get $AG(\mathcal{G})\leq \lceil n^{1/3}\rceil .$ 
\end{theorem}

Let $G=(V,E,(p_{uv})_{(u,v)\in E})$ be the input influence graph. To show Theorem \ref{thm2}, we recall the preliminary notations considered for the proof of Theorem \ref{thm1}, and we give a further preliminary lemma (see the appendix for the proof of the lemma). 
\begin{lemma}\label{lemthm2}
Given a set $U\subseteq V$ of cardinality $h\geq k$, we have $$\sigma(U)\leq \frac{h}{k}\cdot OPT_N(G,k).$$
\end{lemma}

We use Theorem~\ref{lemk} and Lemma~\ref{lemthm2} to show Theorem~\ref{thm2}.

\begin{proof}[Proof of Theorem~\ref{thm2}]
We assume w.l.o.g. that $k>\lceil n^{1/3}\rceil $ and that $OPT_N(G,k)<(\lceil n^{1/3}\rceil)^2$. Indeed, if $k\leq  \lceil n^{1/3}\rceil $, by Theorem \ref{lemk} the claim holds, and if $OPT_N(G,k)\geq (\lceil n^{1/3}\rceil)^2$, then $\frac{OPT_A(G,k)}{OPT_N(G,k)}\leq \frac{|V|}{OPT_N(G,k)}\leq \frac{n}{(\lceil n^{1/3}\rceil)^2}\leq \lceil n^{1/3}\rceil $, and the claim holds as well. 
For any $t\in [k]$, we have
\begin{align}
&k\cdot (OPT_N(G,t)-OPT_N(G,t-1))\nonumber\\
\geq &k\cdot (\E_{L,\rho}[f({\psi}_{\rho,t,L})]-\E_{L,\rho}[f(\psi_{t-1,L})])\nonumber\\
=& \E_{\psi_{t-1,L}}\left[k\cdot \E_{L,\rho}[f({\psi}_{\rho,t,L})-f(\psi_{t-1,L})|\psi_{t-1,L}]\right]\nonumber\\
\geq &\E_{\psi_{t-1,L}}[OPT_A(G,k)-\sigma(R(\psi_{t-1,L}))]\label{eqthm12}\\
=&\E_{\psi_{t-1,L}}[OPT_A(G,k)]-\E_{\psi_{t-1,L}}[\sigma(R(\psi_{t-1,L}))]\nonumber\\
\geq &\E_{\psi_{t-1,L}}[OPT_A(G,k)]-\E_{\psi_{k,L}}[\sigma(R(\psi_{k,L}))]\nonumber\\
\geq &\E_{\psi_{t-1,L}}[OPT_A(G,k)]-\E_{\psi_{k,L}}\left[\frac{|R(\psi_{k,L})|}{k}\cdot OPT_N(G,k)\right]\label{eqthm32}\\
= &OPT_A(G,k)-\frac{\E_{\psi_{k,L}}[|R(\psi_{k,L})|]}{k}\cdot OPT_N(G,k)\nonumber\\
\geq  &OPT_A(G,k)-\frac{\E_{\psi_{k,L}}[|R(\psi_{k,L})|]}{\lceil n^{1/3}\rceil+1}\cdot ((\lceil n^{1/3}\rceil)^2-1)\label{eqthm42}\\
=&OPT_A(G,k)-(\lceil n^{1/3}\rceil-1)\cdot \E_{\psi_{k,L}}[|R(\psi_{k,L})|]\nonumber\\
=&OPT_A(G,k)-(\lceil n^{1/3}\rceil-1)\cdot OPT_N(G,k)\label{eqthm52},
\end{align}
where \eqref{eqthm12} comes from Lemma \ref{lem1}, \eqref{eqthm32} comes from Lemma \ref{lemthm2}, and \eqref{eqthm42} comes from the hypothesis $k>\lceil n^{1/3}\rceil $ and $OPT_N(G,k)<(\lceil n^{1/3}\rceil)^2$. By \eqref{eqthm52}, we get $OPT_N(G,t)-OPT_N(G,t-1)
\geq (OPT_A(G,k)-(\lceil n^{1/3}\rceil-1)\cdot OPT_N(G,k))/k$ for any $t\in [k]$, and by summing such inequality over all $t\in [k]$, we get
\begin{align}
& OPT_N(G,k)\nonumber\\
=&\sum_{t=1}^k (OPT_N(G,t)-OPT_N(G,t-1))\nonumber\\
\geq &\sum_{t=1}^k\frac{OPT_A(G,k)-(\lceil n^{1/3}\rceil-1)\cdot OPT_N(G,k)}{k}\nonumber\\
=&OPT_A(G,k)-(\lceil n^{1/3}\rceil-1)\cdot OPT_N(G,k).\label{eqthm62}
\end{align}
Finally, \eqref{eqthm62} implies that $OPT_A(G,k)\leq \lceil n^{1/3}\rceil \cdot OPT_N(G,k)$, and this shows the claim. 
\end{proof}
\section{Adaptivity Gap for Other Influence Graphs}\label{sec_other}
In this section, we extend the results obtained in Theorem \ref{thm1}, and we get upper bounds on the adaptivity gap of other classes of influence graphs. In particular, we consider the class of {\em $\alpha$-bounded graphs}: a class of undirected graphs parametrized by an integer $\alpha\geq 0$ that includes several known graph topologies. In the following, when we refer to undirected influence graphs, we implicitly assume that, for any undirected edge $\{u,v\}$, there are two directed edges $(u,v)$ and $(v,u)$ having respectively two (possibly) distinct probabilities $p_{uv}$ and $p_{vu}$. 

\subparagraph*{$\alpha$-bounded graphs.} Given an undirected graph $G=(V,E)$ and a node $v\in V$, let $deg_v(G)$ be the degree of node $v$ in graph $G$. Given an integer $\alpha\geq 0$, graph $G$ is an {\em $\alpha$-bounded graph} if $\sum_{v\in V:deg_v(G)>2}deg_v(G)\leq \alpha$, i.e., the sum all the node degrees higher than $2$ is at most $\alpha$. In the following, we exhibit some interesting classes of $\alpha$-bounded graphs:
\begin{itemize}
\item the set of $0$-bounded graphs is made of all the graphs $G$ such that each connected component of $G$ is either an undirected path, or an undirected cycle;
\item if graph $G$ is homeomorphic to a star with $h$ edges, then $G$ is a $h$-bounded graph;
\item if graph $G$ is homeomorphic to a parallel-link graph with $h$ edges, then $G$ is a $2h$-bounded graph;
\item if graph $G$ is homeomorphic to a cycle with $h$ chords, then $G$ is a $6h$-bounded graph;
\item if graph $G$ is homeomorphic to a clique with $h$ nodes, then $G$ is a $h(h-1)$-bounded graph. 
\end{itemize}

In the following, we provide an upper bound on the adaptivity gap of $\alpha$-bounded influence graphs for any $\alpha\geq 0$. 
\begin{theorem}\label{thmlast}
Given $\alpha\geq 0$, let $\mathcal{G}$ be the class of  $\alpha$-bounded influence graphs. Then $$AG(\mathcal{G},k)\leq \min\left\{k,\frac{\alpha}{k}+2+ \frac{1}{1-(1-1/k)^k}\right\}\leq \frac{\sqrt{4(e-1)^2\alpha+(3e-2)^2}+3e-2}{2(e-1)}$$ 
for any $k\geq 2$, i.e., $AG(\mathcal{G})\leq \sqrt{\alpha}+O(1)$.
\end{theorem}

Let $G=(V=[n],E,(p_{uv})_{(u,v)\in E})$ be an $\alpha$-bounded influence graph, and we recall the preliminary notations from Theorem \ref{thm1}. The proof of Theorem \ref{thmlast} is a non-trivial generalization of Theorem \ref{thm1}. In particular, the proof resorts to Theorem \ref{lemk} to get the upper bound of $k$, and, by following the approach of Theorem \ref{thm1}, the following technical lemma is used in place of Lemma \ref{lem2} to get the final upper bound. 
\begin{lemma}\label{lemlast}
When the input influence graph $G$ is an $\alpha$-bounded graph with $\alpha\geq 0$, we have that 
\begin{equation*}
\sigma(R(\psi_{t-1,L}))\leq f(\psi_{t-1,L})+\left(\frac{\alpha}{k}+2\right)\cdot OPT_N(G,k),
\end{equation*}
for any $t\in [k]$ and live-edge graph $L$.
\end{lemma}
\begin{proof}
Given a subset $U\subseteq V$, let $\partial U:=\{u\in U: \exists (u,v)\in E,v\notin U\}$. We have that $\sigma(R(\psi))\leq |R(\psi)|+\sigma(\partial R(\psi))=f(\psi)+\sigma(\partial R(\psi))$ for any partial realisation $\psi$. Thus, to show the claim, it suffices to show that 
\begin{equation*}
\sigma(\partial R(\psi_{t-1,L}))\leq \left(\frac{\alpha}{k}+2\right)\cdot OPT_N(G,k).
\end{equation*}

Let $U\subseteq V$ such that $U$ has at most $k$ connected components. Let $A$ be the set of connected components containing at least one node of degree higher than $2$, and let $B$ be the set of the remaining components, i.e., containing nodes with degree in $[2]_0$ only. By definition of $A$ and $B$, we necessarily have that $|\partial A|\leq \sum_{v\in V:deg_v(G)>2}deg_v(G)\leq \alpha$ and $|\partial B|\leq 2k$. Thus $|\partial U|\leq |\partial A|+|\partial B|\leq \alpha+2k$, and the next claim follows.
\begin{claim}\label{lastclaim}
Given a subset $U\subseteq V$ made of at most $k$ connected components, then $|\partial U|\leq \alpha+2k$. 
\end{claim}
Now, we have that
\begin{align}
&\sigma(\partial R(\psi_{t-1,L}))\nonumber\\
\leq & \sigma(\partial R(\psi_{k,L}))\nonumber\\
\leq & \frac{|\partial R(\psi_{k,L})|}{k}\cdot OPT_N(G,k)\label{lastlem_eq2}\\
\leq & \frac{\alpha+2k}{k}\cdot OPT_N(G,k),\label{lastlem_eq3}
\end{align}
where \eqref{lastlem_eq2} comes from Lemma \ref{lemthm2}, and \eqref{lastlem_eq3} holds since $R(\psi_{k,L})$ contains at most $k$ connected components and because of Claim \ref{lastclaim}. Thus, by \eqref{lastlem_eq3}, the claim of the lemma follows. 
\end{proof}

We can now prove Theorem~\ref{thmlast}.
\begin{proof}[Proof of Theorem~\ref{thmlast}]
For any $t\in [k]$, we have
\begin{align}
&k\cdot (OPT_N(G,t)-OPT_N(G,t-1))\nonumber\\
\geq &k\cdot (\E_{L,\rho}[f({\psi}_{\rho,t,L})]-\E_{L,\rho}[f(\psi_{t-1,L})])\nonumber\\
=& \E_{\psi_{t-1,L}}\left[k\cdot \E_{L,\rho}[f({\psi}_{\rho,t,L})-f(\psi_{t-1,L})|\psi_{t-1,L}]\right]\nonumber\\
\geq &\E_{\psi_{t-1,L}}[OPT_A(G,k)-\sigma(R(\psi_{t-1,L}))]\label{eqthm1_last}\\
\geq &\E_{\psi_{t-1,L}}\left[OPT_A(G,k)-f(\psi_{t-1,L})-\left(\frac{\alpha}{k}+2\right)\cdot OPT_N(G,k)\right]\label{eqthm3_last}\\
=& \E_{\psi_{t-1,L}}[OPT_A(G,k)]-\E_{\psi_{t-1,L}}[f(\psi_{t-1,L})]-\left(\frac{\alpha}{k}+2\right)\cdot \E_{\psi_{t-1,L}}\left[OPT_N(G,k)\right]\nonumber\\
= &OPT_A(G,k)-\sigma(S_{t-1})-\left(\frac{\alpha}{k}+2\right)\cdot OPT_N(G,k)\nonumber\\
= &OPT_A(G,k)-\left(\frac{\alpha}{k}+2\right)\cdot OPT_N(G,k)-OPT_N(G,t-1)\label{eqthm4_last},
\end{align}
where \eqref{eqthm1_last} comes from Lemma \ref{lem1} and \eqref{eqthm3_last} comes from Lemma \ref{lemlast}. Thus, by \eqref{eqthm4_last}, we get the following recursive relation:
\begin{equation}\label{fundeqthm_last}
OPT_N(G,t)\geq \frac{1}{k}\cdot \left(OPT_A(G,k)-\left(\frac{\alpha}{k}+2\right)\cdot OPT_N(G,k)\right)+\left(1-\frac{1}{k}\right)\cdot OPT_N(G,t-1),
\end{equation}
for any $t\in [k]$. By applying iteratively \eqref{fundeqthm_last}, we get
\begin{align*}
&OPT_N(G,k)\\
\geq &\frac{1}{k}\cdot \left(OPT_A(G,k)-\left(\frac{\alpha}{k}+2\right)\cdot OPT_N(G,k)\right)\cdot \sum_{t=0}^{k-1} \left(1-\frac{1}{k}\right)^j\\
= &\left(OPT_A(G,k)-\left(\frac{\alpha}{k}+2\right)\cdot OPT_N(G,k)\right)\cdot \left(1-\left(1-\frac{1}{k}\right)^k\right),
\end{align*}
that, after some manipulations, leads to 
\begin{equation}\label{semifinal_bound}
\frac{OPT_A(G,k)}{OPT_N(G,k)}\leq \frac{\alpha}{k}+2+ \frac{1}{1-(1-1/k)^k} \leq \frac{\alpha}{k}+2+ \frac{1}{1-e^{-1}}~.
\end{equation}
By Theorem \ref{lemk}, we have that $\frac{OPT_A(G,k)}{OPT_N(G,k)}\leq k$, thus, by \eqref{semifinal_bound}, we get 
\begin{align}
&\frac{OPT_A(G,k)}{OPT_N(G,k)}\nonumber\\
\leq & \min\left\{k, \frac{\alpha}{k}+2+ \frac{1}{1-(1-1/k)^k} \  \right\}\nonumber\\
\leq & \min\left\{k, \frac{\alpha}{k}+2+ \frac{1}{1-e^{-1}} \  \right\}\nonumber\\
\leq  & \frac{\sqrt{4(e-1)^2\alpha+(3e-2)^2}+3e-2}{2(e-1)}\label{last_eqq},
\end{align}
where \eqref{last_eqq} is equal to the real value of $k\geq 0$ such that $k=\frac{\alpha}{k}+2+ \frac{1}{1-e^{-1}}$. By \eqref{last_eqq} the claim follows. 
\end{proof}

For the particular case of $0$-bounded influence graphs, the following theorem provides a better  upper bound on the adaptivity gap (the proof is analogue to that of Theorem \ref{thm1}, and is  deferred to the appendix). 

\begin{theorem}\label{thm_0bou}
Let $\mathcal{G}$ be the class of $0$-bounded influence graphs. Then
$$AG(\mathcal{G},k)\leq\min\left\{k,\frac{3}{1-(\max\{0,1-3/k\})^k}\right\}\leq \frac{3e^3}{e^{3}-1}\approx 3.16,\ \forall k\geq 2.$$
\end{theorem}

\section{Future Works}\label{sec_future}
The first problem that is left open by our results is the gap between the constant lower bound provided by Chen and Peng~\cite{Chen2019} and our upper bound on the adaptivity gap for general graphs.
Besides trying to lower the upper bound, a possible direction could be that of increasing the lower bound by finding instances with a non constant adaptivity gap. Since the lower bound given in~\cite{Chen2019} holds even when the graph is a directed path, one direction could be to exploit different graph topologies. 

Although in this work we have improved the upper bound on the adaptivity gap of in-arborescence, there is still a gap between upper and lower bound, thus another open problem is to close it. It would be also interesting to find better bounds on the adaptivity gap of other graph classes, like e.g. out-arborescences. A further interesting research direction is to study the adaptivity gap of some graph classes modelling real-world networks, both theoretically and experimentally. 

The study of the adaptive IM problem in the Linear Threshold model is still open,  in terms of both approximation ratio and adaptivity gap. We observe that in this case the objective function is not adaptive submodular in both myopic and full-adoption feedbacks and therefore the greedy approach by Golovin and Krause~\cite{Golovin2011a} cannot be applied in this case.

The techniques introduced in this paper to relate adaptive policies with non-adaptive ones might be useful to find better upper bounds on the adaptivity gaps in different feedback models, like e.g. the myopic one, or in different graph classes. 

\bibliography{ms}

\appendix
\section{Missing Proofs of Section \ref{sec_inarb}}
\subsection{Proof of Lemma \ref{lem2}}
Given a subset $U\subseteq [n]$, let $\partial U:=\{u\in U:\exists (u,v)\in E,v\notin U\}$. We have that $\sigma(R(\psi))\leq |R(\psi)|+\sigma(\partial R(\psi))=f(\psi)+\sigma(\partial R(\psi))$ for any partial realisation $\psi$. Thus, to show the claim, it suffices to show that $\sigma(\partial R(\psi_{t-1,L}))\leq OPT_N(G,t-1)$. For in-arborescences, we have that $|\partial R(\psi_{t-1,L})|\leq |dom(\psi_{t-1,L})|=t-1$, thus $\sigma(\partial R(\psi_{t-1,L}))\leq OPT_N(G,t-1)$.\qed

\section{Missing Proofs of Section \ref{sec_gen}}
\subsection{Proof of Theorem \ref{lemk}}
Let $G=(V=[n],E,(p_{uv})_{(u,v)\in E})$ be an arbitrary influence graph. Let $\pi^*$ be an optimal adaptive policy subject to $|\pi^*|=k$, and let $\psi_{t,\pi^*,L}$ be the partial realization observed when the $t$-th seed node has been selected by Algorithm \ref{ad_alg} with policy $\pi^*$. For any fixed partial realisation $\psi$ and any $t\in [k]$, we have
\begin{align}
&\E_{L}[f(\psi_{t,\pi^*,L})-f(\psi_{t-1,\pi^*,L})|\psi_{t-1,\pi^*,L}=\psi]\nonumber\\
=&\Delta(\pi^*(\psi)|\psi)\nonumber\\
\leq & \Delta(\pi^*(\psi)|\emptyset)\label{eq1lemk}\\
=&\sigma(\{\pi^*(\psi)\})\nonumber\\
\leq & OPT_N(G,1),\label{eq2lemk}
\end{align}
where \eqref{eq1lemk} holds by adaptive submodularity (Claim  \ref{lem0}). Thus, we get
\begin{align}
&OPT_A(G,k)\nonumber\\
=&\E_{L}[f(\psi_{k,\pi^*,L})]\nonumber\\
=&\sum_{t=1}^k \E_{L}[f(\psi_{t,\pi^*,L})-f(\psi_{t-1,\pi^*,L})]\nonumber\\
=&\sum_{t=1}^k \E_{\psi_{t-1,\pi^*,L}}[\E_L[f(\psi_{t,\pi^*,L})-f(\psi_{t-1,\pi^*,L})|\psi_{t-1,\pi^*,L}]]\nonumber\\
\leq &k\cdot\E_{\psi_{t-1,\pi^*,L}}[OPT_N(G,1)]\label{eq3lemk}\\
=&k\cdot OPT_N(G,1)\nonumber\\
\leq &k\cdot OPT_N(G,k),
\end{align}
where \eqref{eq3lemk} comes from \eqref{eq2lemk}, and the claim follows. \qed

\subsection{Proof of Lemma \ref{lemthm2}}
For any $t\in [h]_0$, let $U_t:=\emptyset$ if $t=0$, and $U_t:=U_{t-1}\cup\{i_t\}$, where $$i_t\in\arg\max_{i\in U\setminus U_{t-1}}(\sigma(U_{t-1}\cup\{i\})-\sigma(U_{t-1})).$$ 
We have that $\Delta_t:=\sigma(U_{t})-\sigma(U_{t-1})$ is non-increasing in $t\in [h]$. Indeed, given $t\in [k-1]$, we have that
\begin{align}
&\Delta_{t+1}\nonumber\\
=&\sigma(U_{t+1})-\sigma(U_{t})\nonumber\\
=&\sigma(U_{t}\cup\{i_{t+1}\})-\sigma(U_{t})\nonumber\\
\leq &\sigma(U_{t-1}\cup\{i_{t+1}\})-\sigma(U_{t-1})\label{eq1lem2thm2}\\
\leq &\max_{i\in U\setminus U_{t-1}}(\sigma(U_{t-1}\cup\{i\})-\sigma(U_{t-1}))\nonumber\\
=&\sigma(U_{t-1}\cup\{i_{t}\})-\sigma(U_{t-1})\nonumber\\
=&\Delta_t,\label{eq2lem2thm2}
\end{align}
where \eqref{eq1lem2thm2} holds since $\sigma$ is a submodular set-function (see \cite{Kempe2015a}). Thus, we necessarily have 
\begin{align}
&\frac{\sigma(U)}{h}\nonumber\\
=&\frac{\sum_{t=1}^h\Delta_t}{h}\nonumber\\
\leq &\frac{\sum_{t=1}^k\Delta_t}{k}\label{eq3lem2thm2}\\
=&\frac{\sigma(U_k)}{k}\nonumber\\
\leq &\frac{OPT_N(G,k)}{k},\label{eq4lem2thm2}
\end{align}
where \eqref{eq3lem2thm2} comes from \eqref{eq2lem2thm2}. By \eqref{eq4lem2thm2}, the claim follows. \qed

\section{Missing Proofs of Section \ref{sec_other}}
\subsection{Proof of Theorem \ref{thm_0bou}}
Let $G=(V=[n],E,(p_{uv})_{(u,v)\in E})$ be a $0$-bounded influence graph. We recall the notation from Theorem \ref{thm1}, and we give the following preliminary lemma, whose proof is analogue to that of Lemma \ref{lem2}. 
\begin{lemma}\label{lem_0bou}
When the input influence graph $G$ is a $0$-bounded graph, we have  
\begin{equation}
\sigma(R(\psi_{t-1,L}))\leq f(\psi_{t-1,L})+2\cdot OPT_N(G,t-1),
\end{equation}
for any $t\in [k]$ and live-edge graph $L$.
\end{lemma}
\begin{proof}
As in Lemma \ref{lem2}, we show that $\sigma(\partial R(\psi_{t-1,L}))\leq 2\cdot OPT_N(G,t-1)$. First of all, we assume that $t\geq 2$, otherwise $\sigma(R(\psi_{t-1,L}))$ and the claim holds. By Lemma \ref{lemthm2}, we have that $\sigma(\partial R(\psi_{t-1,L}))\leq \frac{|\partial R(\psi_{t-1,L})|}{t-1}\cdot OPT_N(G,t-1)$. As $G$ is a $0$-bounded influence graph, we have that $|\partial R(\psi_{t-1,L})|\leq 2(t-1)$. By considering the above  inequalities, we get $\sigma(\partial R(\psi_{t-1,L}))\leq \frac{|\partial R(\psi_{t-1,L})|}{t-1}\cdot OPT_N(G,t-1)\leq \frac{2(t-1)}{t-1}\cdot OPT_N(G,t-1)=2\cdot OPT_N(G,t-1)$, and the claim follows. 
\end{proof}
If $k\leq 3$ the claim trivially holds. Thus, we assume that $k>3$. For any $t\in [k]$, we have
\begin{align}
&k\cdot (OPT_N(G,t)-OPT_N(G,t-1))\nonumber\\
=&k\cdot (\sigma(S_t)-\sigma(S_{t-1}))\nonumber\\
=&k\cdot (\E_L[f(\psi_{t,L})]-\E_L[f(\psi_{t-1,L})])\nonumber\\
\geq &k\cdot (\E_{L,\rho}[f({\psi}_{\rho,t,L})]-\E_L[f(\psi_{t-1,L})])\nonumber\\
= &k\cdot (\E_{L,\rho}[f({\psi}_{\rho,t,L})]-\E_{L,\rho}[f(\psi_{t-1,L})])\nonumber\\
=&k\cdot \E_{L,\rho}[f({\psi}_{\rho,t,L})-f(\psi_{t-1,L})]\nonumber\\
=& \E_{\psi_{t-1,L}}\left[k\cdot \E_{L,\rho}[f({\psi}_{\rho,t,L})-f(\psi_{t-1,L})|\psi_{t-1,L}]\right]\nonumber\\
\geq &\E_{\psi_{t-1,L}}[OPT_A(G,k)-\sigma(R(\psi_{t-1,L}))]\label{eqthm1_0bou}\\
\geq &\E_{\psi_{t-1,L}}[OPT_A(G,k)-f(\psi_{t-1,L})-2\cdot OPT_N(G,t-1)]\label{eqthm3_0bou}\\
=& \E_{\psi_{t-1,L}}[OPT_A(G,k)]-\E_{\psi_{t-1,L}}[f(\psi_{t-1,L})]-2\cdot \E_{\psi_{t-1,L}}[OPT_N(G,t-1)]\nonumber\\
= &OPT_A(G,k)-\sigma(S_{t-1})-2\cdot OPT_N(G,t-1)\nonumber\\
= &OPT_A(G,k)-3\cdot OPT_N(G,t-1)\label{eqthm4_0bou},
\end{align}
where \eqref{eqthm1_0bou} comes from Lemma \ref{lem1} and \eqref{eqthm3_0bou} comes from Lemma \ref{lem_0bou}. Thus, by \eqref{eqthm4_0bou}, we get $k\cdot (OPT_N(G,t)-OPT_N(G,t-1))\geq OPT_A(G,k)-3\cdot OPT_N(G,t-1)$, that after some manipulations leads to the following recursive relation:
\begin{equation}\label{fundeqthm_0bou}
OPT_N(G,t)\geq \frac{1}{k}\cdot OPT_A(G,k)+\left(1-\frac{3}{k}\right)\cdot OPT_N(G,t-1),\quad \forall t\in [k].
\end{equation}
By applying iteratively \eqref{fundeqthm_0bou}, we get
\begin{equation*}
OPT_N(G,k)\geq \frac{1}{k}\cdot \sum_{t=0}^{k-1}\left(1-\frac{3}{k}\right)^{t}\cdot OPT_A(G,k)=\frac{1-\left(1-3/k\right)^k}{3}\cdot OPT_A(G,k),
\end{equation*}
that leads to 
\begin{equation}
\frac{OPT_A(G,k)}{OPT_N(G,k)}\leq \frac{3}{1-(1-3/k)^k}\leq \frac{3}{1-e^{-3}},
\end{equation}
and this shows the claim. \qed

\end{document}